\documentclass[aps,pra,twocolumn,showpacs,amsmath,amssymb,nofootinbib,longbibliography
]{revtex4-1}

\usepackage[pdftex]{graphicx}
\usepackage[usenames,dvipsnames]{xcolor}
\usepackage{epstopdf}
\usepackage[hidelinks]{hyperref}

\usepackage[normalem]{ulem}

\graphicspath{{FIGs/}}

 \usepackage[utf8]{inputenc}

\usepackage{beramono}
\usepackage{listings}

\usepackage{lmodern}

\usepackage{amssymb,amsmath,amsthm}

\newtheorem{theorem}{Theorem}
\theoremstyle{definition}
\newtheorem{definition}{Definition}

\theoremstyle{remark}

\newcommand{\patbox}[1]{\vspace{2mm}\noindent\fbox{\parbox{0.47\textwidth}{\hspace{1mm}\parbox{0.45\textwidth}{\hspace{1mm}\\#1\vspace{1.5mm}}}}\\}

\lstdefinelanguage{Julia}%
  {morekeywords={abstract,break,case,catch,const,continue,do,else,elseif,%
      end,export,false,for,function,immutable,import,importall,if,in,%
      macro,module,otherwise,quote,return,switch,true,try,type,typealias,%
      using,while},%
   sensitive=true,%
   alsoother={\$},%
   morecomment=[l]\#,%
   morecomment=[n]{\#=}{=\#},%
   morestring=[s]{"}{"},%
   morestring=[m]{'}{'},%
     literate={é}{{\'e}}1
           {è}{{\`e}}1
           {ù}{{\`u}}1
}[keywords,comments,strings]%

\usepackage{chngcntr}

\begin{document}

\title{Lanczos recursion on a quantum computer for the Green's function and ground state}
\author{Thomas E.~Baker}
\affiliation{Institut quantique \& Département de physique, Université de Sherbrooke, Sherbrooke, Québec J1K 2R1 Canada}
\affiliation{Department of Physics, University of York, Heslington,York YO10 5DD, United Kingdom}

\begin{abstract}
A state-preserving quantum counting algorithm is used to obtain coefficients of a Lanczos recursion from a single ground state wavefunction on the quantum computer.  This is used to compute the continued fraction representation of an interacting Green's function for use in condensed matter, particle physics, and other areas. The wavefunction does not need to be re-prepared at each iteration. The quantum algorithm represents an exponential reduction in memory over known classical methods. An extension of the method to determining the ground state is also discussed. 
\end{abstract}
\maketitle

\section{Introduction}

Quantum computers are able to store a superposition of states on each qubit, and this has led to the demonstration that some algorithms can provide a quantum advantage in terms of the amount of time needed to run some tasks \cite{nielsen2010quantum}.  Finding algorithms that are more efficient than classical equivalents are highly sought, particularly those for quantum chemistry applications \cite{lanyon2010towards}. 

Many existing proposals focus on cases where a measurement from a wavefunction is used to obtain some property of a molecular system; however, the time needed to solve for a wavefunction can be prohibitively long for large systems, as demonstrated by Poulin,~{\it et.~al.} \cite{poulin2014trotter,lemieux2020resource}.  The problem is provably hard \cite{schuch2009computational}, but improved methods must be found to ensure the next generation of technologies are discovered.  Using the fewest measurements to obtain the most descriptive quantities is therefore desirable \cite{BP20}. This is the central concept here.

One useful quantity is the fully interacting Green's function, which is approximated in current strategies on the classical computer such as dynamical mean-field theory (DMFT) \cite{georges1996dynamical,senechal2008introduction}, {\it GW} \cite{aryasetiawan1998gw}, random-phase approximation (RPA) methods \cite{chen2017random}, and quantum field theories \cite{peskin1995quantum}.  The Green's function contains the ground state and excitation energies notably in the spectral function, which is the imaginary part of the Green's function \cite{economou1983green}. 

Recently, an excellent initial step in obtaining the Green's function was to evaluate correlation functions of the form $\langle\phi(t_1)\phi(t_2)\ldots\phi(t_r)\rangle$ for different times $t_i$ and fields $\phi$ by Rall in Ref.~\onlinecite{rall2020quantum}, effectively adapting techniques commonly used for time evolution of matrix product states to the quantum computer \cite{schollwock2005density,baker2019m}. This form is well-suited for non-equilibrium phenomena. However, this form of the correlation function suffers from known errors.  First, a finite time interval will create difficulties when evaluating a Fourier transform.  A well-known uncertainty in the amount of time used to resolve the correlation function, $\Delta t$, and the spread of energies, $\Delta E$, is $\Delta t\Delta E\geq1/2$ \cite{boas2006mathematical}. This implies that even long time intervals will still not be able to determine the energy down to arbitrary accuracy.  In practice, this will cause properties such as peaks in the spectral function to broaden.  Another source of error is the Nyquist error when the function is over- or under-sampled in time.

Another method common to tensor networks are correction-vector techniques which solve for a single frequency of the Green's function \cite{jeckelmann2002ground,schollwock2005density}. An equivalent method suitable for a quantum computer was recently shown in Ref.~\onlinecite{tong2020fast}. However, in general, it should be expected that the wavefunction is not efficiently available \cite{schuch2009computational,poulin2014trotter,lemieux2020resource}, so the multiple wavefunction preparations required  will limit the usefulness of both of these algorithms and others, including those for impurity problems \cite{dallaire2016method,dallaire2016quantum,bauer2016hybrid}.

 Computing the Green's function directly in frequency space for all frequencies at once would be highly advantageous to ensure improved Fourier transforms. A form that is useful in several applications that has this property is the continued fraction representation \cite{senechal2008introduction}. If the continued fraction representation is known, then a Fourier transform can then be applied to transform it back to the time domain.

In this work, the continued fraction representation of the Green's function, as reviewed in Sec.~\ref{lanczos}, is obtained through a quantum version of Lanczos recursion (QLR). For demonstration purposes, the one-body Green's function is found from the operator $\hat c_{j\bar\sigma}^\dagger\hat c_{i\sigma}$, although generalizations to other operators are straight-forward.  A state-preserving quantum counting algorithm \cite{brassard2002quantum,marriott2005quantum,temme2011quantum} is used to extract the necessary coefficients to construct the continued fraction in Sec.~\ref{QLR}. Essentially, the quantum counting algorithm replaces the copy operation on a classical computer \cite{temme2011quantum}. The QLR receives an additional step to obtain the ground state and energy from a knowledge of the Lanczos coefficients, and that is also demonstrated in Sec.~\ref{groundstate}.  Additional  comparison with known classical and quantum algorithms are discussed in Sec.~\ref{comparison}.

\section{Lanczos recursion}\label{lanczos}

Constructing a Krylov subspace is an effective way to solve for the ground state of a quantum Hamiltonian, $\mathcal{H}$, as it is applied on some arbitrary starting state.

\begin{definition}[Krylov subspace]\label{krylovcvg}
Given an initial state $\Phi$ and self-adjoint (Hermitian) operator referred to here as the Hamiltonian $\mathcal{H}$, the set $\{\Phi,\mathcal{H}\Phi,\mathcal{H}^2\Phi,\ldots\mathcal{H}^{N-1}\Phi\}$ of $N$ vectors corresponding to the $N$ states in the full Hilbert space is known as the Krylov subspace.
\end{definition}

The Krylov subspace contains the ground state wavefunction and can be used to find it. To see this, application of $\mathcal{H}$ to some power $p$ to $\Phi$ generates a spectral decomposition of $\Phi$ of the form 
$\mathcal{H}^p|\Phi\rangle=\sum_n(E_n)^p|\phi_n\rangle$ 
in terms of eigenvalues $E_n$ and eigenvectors $\phi_n$ of $\mathcal{H}$. By including repeated application of $\mathcal{H}$ (larger $p$), the most extremal eigenvalues of $\mathcal{H}$ will become further from the next nearest eigenvalues \cite{baker2019m}. Diagonalizing $\mathcal{H}$ projected into this subspace gives the ground state energy, and this is a generally reliable method for solving systems even if only a partial basis is used. A more exhaustive proof is omitted in favor of this intuition.

Constructing the Krylov subspace vectors by orthogonalizing against all vectors previously found is known as a power method.  However, if all vectors are guaranteed to be orthogonal, then a 3-term Lanczos recursion can be used to find the next vector in the subspace \cite{lanczos1950iteration}.

\begin{theorem}[Lanczos recursion]\label{basiclanczosrecursion}
The Krylov subspace is generated by a Lanczos recursion of the form
\begin{equation}\label{lanczosrecursion}
|\psi_{n+1}\rangle=\mathcal{H}|\psi_n\rangle-\alpha_n|\psi_n\rangle-\beta_n|\psi_{n-1}\rangle
\end{equation}
where $n\in\mathbb{Z}^+$, $\alpha_n\equiv\langle\psi_n|\mathcal{H}|\psi_n\rangle$, $\beta^2_n\equiv\langle\psi_{n-1}|\psi_{n-1}\rangle$, $|\psi_{-1}\rangle=0$, and $|\psi_0\rangle=|\Phi\rangle$.
\end{theorem}

\begin{proof}[Proof]
If all previous vectors in the Krylov subspace expansion are orthonormal, then $\langle\psi_m|\mathcal{H}|\psi_n\rangle$ is identically zero if $|m-n|>1$ since $\mathcal{H}|\psi_n\rangle$ is contained in the Krylov subspace and $|\psi_m\rangle$ is made to be orthogonal. Note that $\alpha_n$ is not necessarily an eigenvalue of $\mathcal{H}$.
\end{proof}

From the coefficients for the Lanczos recursion, the Green's function can be constructed \cite{senechal2008introduction}. There exist many types of Green's functions \cite{AMST}, but the focus is on one-body Green's function without loss of generality.

\begin{definition}[One-body Green's function]
A one-body Green's function is defined as \cite{AMST}
\begin{equation}
\mathcal{G}_{j\bar\sigma;i\sigma}(\omega)\equiv\langle\Psi|\hat c_{j\bar\sigma}^\dagger\;(\omega-\mathcal{H}\pm i\eta)^{-1}\;\hat c_{i\sigma}|\Psi\rangle
\end{equation}
in the limit as $\eta\rightarrow0$ ($\eta\in\mathbb{R}$) giving the advanced $(-)$ or retarded $(+)$ Green's function. The one-body Green's function is given in terms of raising and lower operators $\hat c^{(\dagger)}_{i\sigma}$ for site $i$ ($j$) and spin $\sigma$ ($\bar\sigma$) and ground state wavefunction $\Psi$ at frequency $\omega$. Only subscripted values $i$ are integer indices, otherwise $\sqrt{-1}\equiv i$.
\end{definition}

\begin{theorem}[Continued fraction representation of the Green's function]
The continued fraction representation of the Green's function, $\mathcal{G}$, of a Hamiltonian $\mathcal{H}$ for non-relativistic systems is
\begin{equation}\label{continuedfrac}
\mathcal{G}_{j\bar\sigma;i\sigma}(\omega)
=\cfrac{\langle\Psi|\hat c^\dagger_{j\bar\sigma}\hat c_{i\sigma}|\Psi\rangle}{\omega-\alpha_0-\cfrac{\beta_1^2}{\omega-\alpha_1-\cfrac{\beta_2^2}
{\ddots}
}}
\end{equation}
where $\alpha_n$ and $\beta_n$ coefficients are defined by a Lanczos recursion algorithm from Eq.~\eqref{lanczosrecursion}
\end{theorem}

\begin{proof}[Proof]
Begin with the tridiagonal representation of the Hamiltonian in terms of the Lanczos coefficients by inspection of Eq.~\eqref{lanczosrecursion},
\begin{equation}\label{LanczosHam}
\mathcal{H}=\left(\begin{array}{cccccc}
\alpha_0&\beta_1&0&0&\cdots&0\\
\beta_1&\alpha_1&\beta_2&0&\cdots&0\\
0&\beta_2&\alpha_2&\beta_3&\cdots&0\\
0&0&\beta_3&\alpha_3&\ddots&0\\
\vdots &\vdots & \ddots & \ddots & \ddots &\beta_N\\
0&0&0&0&\beta_N&\alpha_N
\end{array}\right)
\end{equation}
for $N$ iterations of the Lanczos algorithm expressed in the generated set of $\{|\psi_n\rangle\}$ vectors.  Note further that $\omega-\mathcal{H}$ can be expressed by subtracting the Hamiltonian from an identity matrix times $\omega$.  

The task is then find the inverse, where it is noted that an effective 2$\times$2 super matrix can be formed by first isolating only the upper left entry in Eq.~\eqref{LanczosHam} with $\alpha_0$ as
\begin{equation}
\omega-\mathcal{H}=\left(\begin{array}{cc}
\omega-\alpha_0&\beta_1\\
\beta_1&\omega-\mathcal{H}_1
\end{array}\right)
\end{equation}
where the lower right $(N-1)\times (N-1)$ block is a single matrix $\mathcal{H}_1$.  The upper left element denoted by (1,1) of the inverse of $\omega-\mathcal{H}$ is then
\begin{align}
\left[\left(\omega-\mathcal{H}\right)^{-1}\right]_{11}
&=\frac1{(\omega-\alpha_0)-\frac{\beta_1^2}{\omega-\mathcal{H}_1}}
\end{align}
which can be repeated iteratively starting from $\mathcal{H}_1$ to generate Eq.~\eqref{continuedfrac} with $\omega\rightarrow\omega\pm i\eta$.
\end{proof}

 Here, division by a matrix is synonymous with the inverse in all cases since the coefficients $\alpha_n$, $\omega$, and $\beta_n$ are all scalar values. The continued fraction can be evaluated out to a certain order where the coefficients are either small enough or large enough to justify leaving off the next level of the continued fraction \cite{demir2006coulomb}. For example if $\beta_n$ becomes too large or too small, the $n^\mathrm{th}$ level or $(n-1)^\mathrm{th}$ level is well approximated by zero. It is possible to construct an extrapolation of known coefficients \cite{viswanath2008recursion}.

The operator used here $\hat c^\dagger_{j\bar\sigma}\hat c_{i\sigma}$ is to more clearly identify with the common notation in the literature for a 2-point correlation function \cite{AMST}. The more general expression would be for some Hermitian operator $\hat\Omega$, noting that $\hat c^\dagger_{j\bar\sigma}\hat c_{i\sigma}$ is not Hermitian if $i\neq j$ and $\bar\sigma\neq\sigma$.  If a Hermitian $\hat\Omega$ is required in place of $\hat c^\dagger_{j\bar\sigma}\hat c_{i\sigma}$, then both of $(\hat c^\dagger_{j\bar\sigma}\hat c_{i\sigma}\pm\hat c^\dagger_{i\sigma}\hat c_{j\bar\sigma})$ can be computed in Eq.~\eqref{continuedfrac} and added or subtracted. Equally, other operators such as the higher order correlation functions or pairing terms, for example $(\hat c^\dagger_{j\bar\sigma}\hat c^\dagger_{i\sigma}\pm\hat c_{i\sigma}\hat c_{j\bar\sigma})$, can be found.  Note also that a representation in a basis (see discussion in Ref.~\onlinecite{BP20,bakerPRB18}) is used but operators can equally be expressed in real space or in momentum space.

It can also be noted that $\hat c_{i\sigma}$ (or $\hat c^\dagger_{i\sigma}$) is not Hermitian.  One way to force a Hermitian operator is to require two computations: once for $\hat\Omega_+=\hat c^\dagger_{i\sigma}+\hat c_{i\sigma}$ and a second time for $\hat\Omega_-=i(\hat c^\dagger_{i\sigma}-\hat c_{i\sigma})$ and then noting that $2c^\dagger_{i\sigma}=\hat\Omega_+-i\hat\Omega_-$ and $2c_{i\sigma}=\hat\Omega_++i\hat\Omega_-$. The resulting coefficients can be combined according to these rules to generate results for $\hat c_{i\sigma}$ or $\hat c^\dagger_{i\sigma}$ if some procedure cannot directly apply $\hat c_{i\sigma}$.

\begin{figure}[b]
\includegraphics[width=\columnwidth]{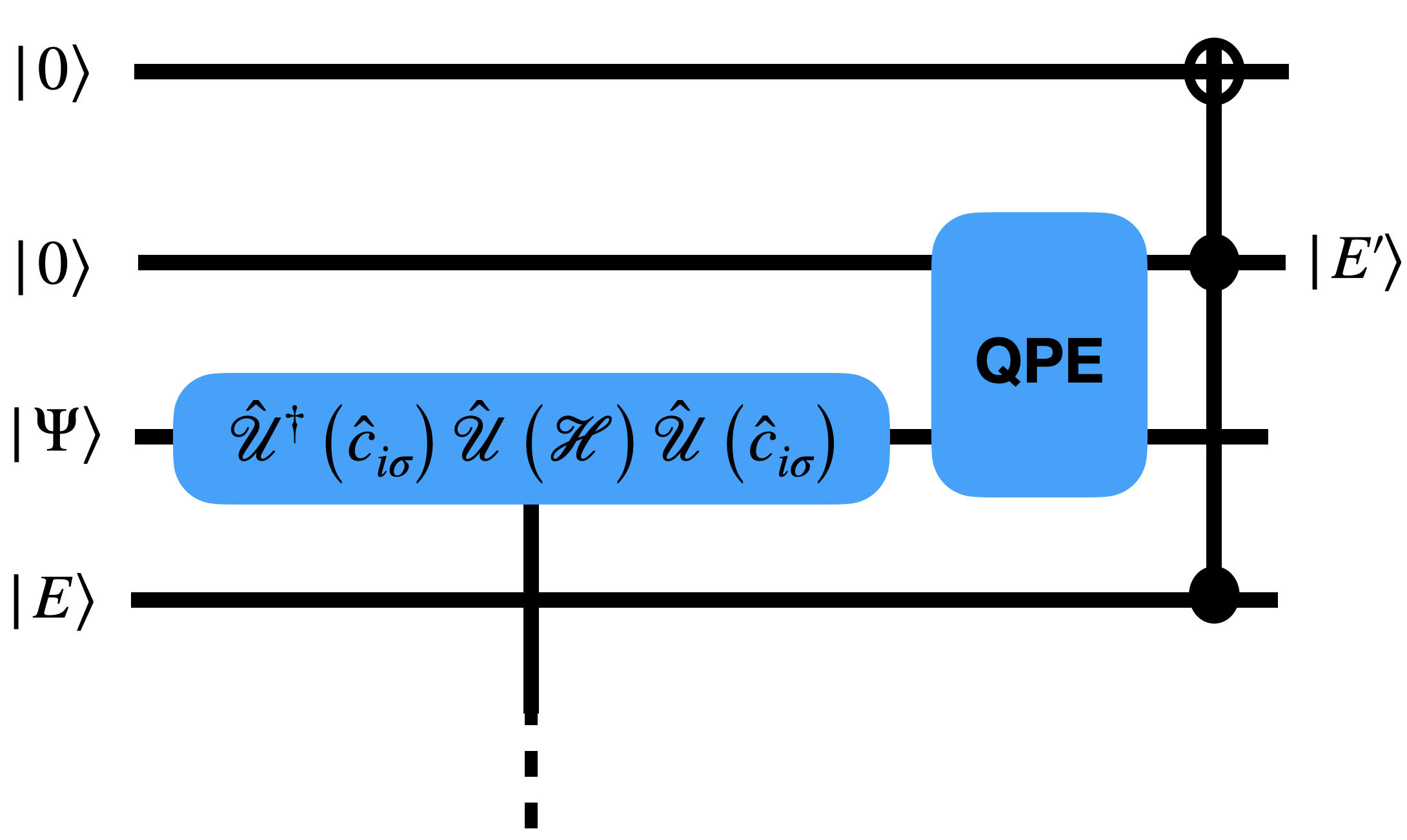}
\caption{\label{Lanczoscounting} 
Modified quantum counting circuit to obtain $\alpha_0$ from input state $\Psi$ with energy $E$.  Vertical lines that descend below the figure account for all auxiliary qubits and recursive steps for $n>0$ including the previous $\alpha_n$ and $\beta_n$ in Eq.~\eqref{lanczosrecursion}, and the application of $\hat c_{i\sigma}$ can be made to be Hermitian operator (see text). Quantum phase estimation (QPE) generates another reference energy $E'$ from the new state. The last step is a comparison of $E'$ and the reference energy $E$ that is converted to a single qubit.}
\end{figure}

\section{Quantum Lanczos recursion}\label{QLR}

QLR can determine the coefficients $\alpha_n$ and $\beta_n$ on the quantum computer, without destroying the initial wavefunction $\Psi$ \cite{temme2011quantum,BP20}.  Both $\alpha_n$ and $\beta_n$ are classical variables throughout and the general strategy is represented in the circuit diagram of Fig.~\ref{Lanczoscounting} and a step-by-step analysis of the operations required are presented in Apx.~\ref{lanczoscoeffex}.

The method known as state-preserving quantum counting (related to quantum amplitude estimation \cite{brassard2002quantum,temme2011quantum,BP20} and Apx.~\ref{history}) is used to find the coefficients $\alpha_n$ and $\beta_n$. The basic strategy of quantum counting is to store the energy of the initial wavefunction, $E$, and compare against the energy after all operations are applied, $E'$.   If the same final state is recovered ($E=E'$) then the count is accepted, and if a recovery procedure is required ($E\neq E'$) to re-obtain the original wavefunction then the count is rejected \cite{marriott2005quantum}.  The ratio of acceptances to total iterations is the expectation value up to a normalization. 

The ground state is given as $\Psi$. The next task to apply a unitary $\hat{\mathcal{U}}(\hat c_{i\sigma})$ which can be equivalent to $\hat\Omega$ from earlier. For the general case, the operator $\hat{\mathcal{U}}$ will always be taken the represent the operator expressed in the argument. The resulting state is $|\Phi\rangle=\hat{\mathcal{U}}(\hat c_{i\sigma})|\Psi\rangle$.  Any errors from applying operators are not considered as in Ref.~\onlinecite{temme2011quantum}.

The counting procedure is illustrated on the first coefficient $\alpha_0$ and then generalized. Applying $\mathcal{H}$ to $|\psi_0\rangle=|\Phi\rangle$ as could be done as in Ref.~\onlinecite{low2019hamiltonian} to give
\begin{equation}\label{HamPsi}
|\mathcal{H}\psi_0\rangle\propto\hat{\mathcal{U}}(\mathcal{H})|\psi_0\rangle=\alpha_0|\psi_0\rangle+\alpha^\perp|\psi_0^\perp\rangle
\end{equation}
where all perpendicular signs ($\perp$) signify a set of states that are all orthogonal to the target state.  For now, the target state is $|\psi_0\rangle$. Since the state $|\psi_0\rangle$ is not an eigenvector of $\mathcal{H}$, note that $|\psi_0\rangle$ must be reverted back to $\Psi$ in order to apply quantum phase estimation (QPE) and find $E'$ since QPE acts on eigenvectors \cite{nielsen2010quantum}.

To do this, a unitary $\hat{\mathcal{U}}^\dagger(\hat c_{i\sigma})$ (which is equal to the unitary representation of $\hat c_{i\sigma}^\dagger$ for this example) is then applied to $|\psi_0\rangle$. The state is then $\alpha_0|\Psi\rangle+\alpha^\perp|\Psi^\perp\rangle$ and suitable for QPE.   Importantly, the leading coefficient is equal to $\alpha_0$ here. This is because all operations are unitary, so $\Psi^\perp$ remains orthogonal to $\Psi$ since unitary operations maintain orthogonality. Thus,  the coefficient $\alpha_0$ in Eq.~\eqref{HamPsi} acting on $|\psi_0\rangle$ and $|\Psi\rangle$ are the same. 

QPE is now used to obtain the ground state energy, $E'$.  The energy is in superposition with the value found for the other perpendicular states, $\Psi^\perp$.  The energies $E$ and $E'$ are compared with a series of CNOT gates to obtain a value of 0 or 1 (accept or reject) on a single pointer qubit \cite{nielsen2010quantum,brassard2002quantum,temme2011quantum}. The pointer qubit is measured.  If the first $d$-bits of the energy match (the other bits are ignored since they are affected by measurement), the value is accepted and $\Psi$ was recovered.  The restriction on $d$ is that the accepted difference $E'-E$ is less than the difference between $E$ and the first excited state \cite{temme2011quantum}.

If the reject outcome is found, then some state in $\Psi^\perp$ was measured and then a recovery procedure is performed until the original wavefunction is obtained, which is explained in Refs.~\onlinecite{brassard2002quantum}, \onlinecite{temme2011quantum}, and Apx.~\ref{lanczoscoeffex}.   The ratio of successes to the total number of iterations is $\alpha_0$ up to a chosen normalization factor when $\hat{\mathcal{U}}$ was determined. The coefficients can be output to a classical user and put back onto another register as a classical number or incorporated in the next operator directly. On the circuit diagram, they are given as another register.

In general ($n\geq0$), an operator $\hat G_n$ can be defined to implement the Lanczos recursion from Thm.~\ref{basiclanczosrecursion} as
\begin{equation}\label{Gdef}
|\psi_n\rangle=\hat G_n|\Psi\rangle
\end{equation}
where there is an assumed dependence of $\hat G_n$ on $\hat\Omega$ and all coefficients indexed by $m$ for $m\leq n$ of a chosen $n$ that are used to determine $\hat G_n$.  The first two operators are
\begin{align}
\hat G_1&=\hat{\mathcal{U}}\left(\mathcal{H}-\alpha_0\right)\hat{\mathcal{U}}(\hat c_{i\sigma})=\hat{\mathcal{U}}\left(\mathcal{H}-\alpha_0\right)\hat \Omega\\
\hat G_2&=\hat{\mathcal{U}}(\left(\mathcal{H}-\alpha_1\right)\left(\mathcal{H}-\alpha_0\right)-\beta_1)\hat \Omega
\end{align}
as can be seen from Apx.~\ref{lanczoscoeffex}.  Defining
\begin{equation}\label{alphaG}
\alpha_n\propto\langle\Psi|\hat G^\dagger_n\;\hat{\mathcal{U}}(\mathcal{H})\,\hat G_n|\Psi\rangle
\end{equation}
demonstrates how the state-preserving quantum counting algorithm can be used by applying the operator $(\hat G^\dagger_n\;\hat{\mathcal{U}}(\mathcal{H})\,\hat G_n)$ onto $|\Psi\rangle$. A proportionality is used here as in Apx.~\ref{lanczoscoeffex}, because of an overall normalization factor.

Whether there is an advantage in applying the operators individually with $\hat{\mathcal{U}}$ operators or as one single operator with $\hat G_n$ operators is problem and implementation dependent. The combined operators $\hat G_n$ are more compact, so the presentation here will use that form. This is highly dependent on what combination of unitaries is used to represent the operator for a given problem, and so will not be discussed here.

In order to find the $\beta_n$ coefficients, note that
\begin{equation}\label{nextbeta}
\beta_n=\langle\psi_{n-1}|\mathcal{H}|\psi_n\rangle=\langle\psi_n|\mathcal{H}|\psi_{n-1}\rangle
\end{equation}
by direct evaluation of Eq.~\eqref{lanczosrecursion}. This term can also be estimated with quantum counting.  Instead of seeking $|\psi_n\rangle$ from $|\mathcal{H}\psi_n\rangle$ as in Eq.~\eqref{alphaG}, $|\psi_{n-1}\rangle$ is sought,
\begin{equation}
\beta_n\propto\langle\Psi|\hat G^\dagger_{n-1}\;\hat{\mathcal{U}}(\mathcal{H})\,\hat G_n|\Psi\rangle
\end{equation}
or the Hermitian conjugate of the operator. Similarly to the $\alpha_n$ coefficient, the operator to count $\beta_n$ from $\Psi$ can be found by repeated application of $(\hat G^\dagger_{n-1}\;\hat{\mathcal{U}}(\mathcal{H})\,\hat G_n)$.

One cycle of the full algorithm is:\\
\patbox{
\textbf{Quantum Lanczos recursion: Algorithm for the Green's function}

\begin{enumerate}
 \item Obtain the ground state wavefunction $\Psi$. 
 \item Apply $\hat c_{i\sigma}$ to $\Psi$ giving $|\psi_0\rangle=\hat c_{i\sigma}|\Psi\rangle$. 
 \item Start iteration $n$, starting at $n=0$.
 \item Compute the $n$th vector from Thm.~\ref{basiclanczosrecursion}. 
 \item Use quantum counting to find $\alpha_n=\langle\psi_n|\mathcal{H}|\psi_n\rangle$. 
 \item Find $\beta_n$ from Eq.~\eqref{nextbeta} with quantum counting.
 \item Repeat from Step 3, incrementing $n$ for a certain number of iterations or until some criterion (for $\beta_n$ or the energy) is reached.  
 \item Use quantum counting for $\langle\Psi|\hat c_{j\bar\sigma}^\dagger\hat c_{i\sigma}|\Psi\rangle$.
\end{enumerate}
}

A stopping criteria can be applied either when the $\beta_n$ coefficients become less than a defined tolerance or a specified number of iterations is accomplished. Note also that, solving Eq.~\eqref{LanczosHam} shows the convergence in energy, which can be another criteria for determining how many levels of the continued fraction are sufficient.

In principle, one could substitute the counting for measurement (as may be useful for near-term investigations), but since $\Psi$ is preserved at the end of this computation, it is ready without re-preparation \cite{temme2011quantum}.  Another coefficient for the continued fraction or component of the Green's function can be found. Alternatively, another $\Psi$ close-by can be solved, similar to the strategy introduced in the recycled wavefunction for minimal prefactor (RWMP) methods from Ref.~\onlinecite{BP20}. This method could also be applied in other cases of interest (for example, relativistic few-body quantum physics \cite{demir2006coulomb}) or different recursion relations with different coefficients \cite{boas2006mathematical}.  

\section{Lanczos for the ground states}\label{groundstate}

The QLR could be used to obtain the ground state.  Consider that each element of the Krylov subspace is of the form $\hat G_n|\Psi\rangle$ as in Eq.~\eqref{Gdef}, with $\hat G_n$ containing some combination of Hamiltonians and Lanczos coefficients $\alpha_n$ and $\beta_n$. In order to do this, it is necessary to represent $\hat G_n$ as a linear combination of unitaries and use a technique to apply the operator \cite{kothari2014efficient}, just as before.

The starting initial state $|\Xi\rangle$ is a known eigenstate of some base Hamiltonian, $\mathcal{H}_0$. This is the same starting assumption that is used for real time evolution (RTE) in that some eigenstate is prepared on the quantum computer before any operations are performed. The Lanczos coefficients are generated just as before, but $\mathcal{H}_0$ is used in the QPE step.  Once all Lanczos coefficients are determined, the tridiagonal Hamiltonian from Eq.~\eqref{LanczosHam} is diagonalized on the classical computer \cite{press1992numerical} and the weights of the ground state from the Krylov basis are determined.  The weights of the lowest lying eigenvector of Eq.~\eqref{LanczosHam} are denoted as $\gamma_n$ and defined as
\begin{equation}
|\Psi\rangle=\sum_n\gamma_n|\psi_n\rangle
\end{equation}
where $|\psi_n\rangle$ is now determined from the starting state $|\Phi\rangle=|\Xi\rangle$ instead of $|\Psi\rangle$ as was the case for the Green's function. To express the operator needed to apply to the state $|\Xi\rangle$, the $\gamma_n$ coefficients are used in combination with the $\hat G_n$ operators to form
\begin{equation}\label{groundstateY}
\hat Y \equiv \sum_n\gamma_n\hat G_n
\end{equation}
and when applied on $|\Xi\rangle$, gives the ground state of the quantum problem $|\Psi\rangle$,
\begin{equation}
\hat Y|\Xi\rangle=|\Psi\rangle
\end{equation}
or is approximately equal when using a truncated Krylov basis.  Note that a linear combination of a linear combination of operators is still just a linear combination of operators, so the same tools to apply operators used previously apply here.   Since the ground state energy is also determined when diagonalizing Eq.~\eqref{LanczosHam}, both the ground state and energy are available for the QLR for the Green's function without an extra QPE at the start.  This also provides a means to determine convergence, which must be sufficient for anticipated operators acting on the resulting wavefunction. For example, QPE requires a certain level of precision to obtain good eigenvalues \cite{nielsen2010quantum}.

Note that only the QPE to determine $E$ and $E'$ uses the base Hamilonian $\mathcal{H}_0$.  The Hamiltonian used in the recursion relations ($\hat G_n$ operators) is of the form
\begin{equation}\label{tuneH}
\mathcal{H}(\lambda)=\mathcal{H}_0+\lambda\mathcal{H}_1
\end{equation}
for some tuning parameter $\lambda$. This strategy is similar to the form used for RTE \cite{BP20}, and necessary when the number of Lanczos coefficients is not sufficient to obtain an accurate ground state energy.  Hence, a series of Hamiltonians can be chained together.  Lanczos is known to be rapidly convergent \cite{senechal2008introduction}, so the $\lambda$-step required may be larger than RTE since the adiabatic evolution of a wavefunction requires a small time step \cite{poulin2014trotter}.

The full algorithm for the ground state is:\\
\patbox{
\textbf{Quantum Lanczos recursion: Algorithm for the ground state}

\begin{enumerate}
 \item A starting state $|\Xi\rangle$ is chosen with some base Hamiltonian, $\mathcal{H}_0$.
 \item Start a counter at $n=0$. Determine $\alpha_0$.
 \item Construct the operator $\hat G_n$ representing the $n$th term in Thm.~\ref{basiclanczosrecursion} with the Hamiltonian from Eq.~\eqref{tuneH}.
 \item $\hat G_n$ is applied onto the current state.
 \item Quantum counting is used to determine the next Lanczos coefficients.
 \item With the new coefficients, the algorithm returns to step 3, incrementing $n$ by one.
 \item The Hamiltonian from Eq.~\eqref{LanczosHam} can be diagonalized to give the coefficients of the ground state eigenvector $\{\gamma_n\}$ and energy. The operator $\hat Y$ from Eq.~\eqref{groundstateY} is formed and applied on state $|\Xi\rangle$.
 \item (optional) A new Hamitonian can be solved from the new ground state wavefunction.
\end{enumerate}
}

Some additional detail on known methods to apply the operators is provided in Apx.~\ref{oblivious}. The determination of the $\gamma_n$ coefficients could be done on either the quantum or classical computer, albeit with extra qubit overhead on the quantum computer. Since $\alpha_n$ and $\beta_n$ are classical variables and output to the classical user in the simplest case, the matrix can be directly represented and diagonalized.

Note that solving for ground states is known to be QMA-complete in general \cite{schuch2009computational} and that approximations of the expectation values in the quantum counting and the use of the truncated Krylov basis appear in the QLR. The largest complexity issue is still the same as RTE where the time step needed in general is too complicated to know in advance for every problem \cite{BP20}.  


 

\section{Relation to other methods}\label{comparison}

Previously known methods are now compared and contrasted on both the classical and quantum computers with the QLR method. 

\subsection{Comparison with classical techniques}\label{classicalcomp}

On the classical computer, it becomes difficult to construct a continued fraction representation that extends beyond 10 or so levels.  Exact diagonalization \cite{wietek2018sublattice} gives the Lanczos coefficients to numerical precision, but since the system sizes are limited by memory, only a few floors are possible.

The computation of continued fractions has been explored with tensor networks \cite{schollwock2005density,kuhner1999dynamical}; however, it has been found that the numerical truncation procedure (singular value decomposition \cite{baker2019m}) decreases the normalization of the wavefunction ansatz and the reachable number of accurate floors is nearly the same as exact diagonalization \cite{foleyreseaux}. Hence, the continued fraction representation is not as useful on a classical computer. More generally, the issue of numerical stability is an issue for Lanczos techniques on the classical computer \cite{cahill2000numerical}.

Because the quantum wavefunction will not be affected by the same limits as the classical algorithms (namely, the quantum wavefunction does not grow as more operations are performed and results are accurate out to quantum precision), the quantum advantage here is an exponential reduction in memory \cite{nielsen2010quantum} (see also the introduction relating to the extended Church-Turing thesis in Ref.~\onlinecite{aaronson2011computational}). Surpassing the limitations given by numerical instabilities of Lanczos on the classical computer means that even small systems will see a quantum advantage.

\subsection{Comparison with quantum techniques}\label{quantumcomp}

While the costs of the quantum computing algorithm are drastically cheaper than the classical equivalent algorithms, comparison with other quantum algorithms is now discussed here, with particular attention paid to Refs.~\onlinecite{rall2020quantum,tong2020fast,dallaire2016method,dallaire2016quantum,bauer2016hybrid}. Each of these relies on multiple measurements of the wavefunction, and therefore incur expensive wavefunction preparation time.  To formalize in the big-$O$ notation, the RWMP methods \cite{BP20} (including the QLR here) have $O(1)$ queries for the wavefunction and this is less than other algorithms.  A more rigorous comparison for doing fewer wavefunction preparations compared to fewer oracle queries of relevant operators is not possible to be conducted in general.  However, it is expected that wavefunction preparation will be far more costly with RTE, so the difference between this method and other algorithms can be weeks, months, or decades \cite{poulin2014trotter,lemieux2020resource}. 

In regards to Green's function, the discussion throughout this paper is done with regards to general Hamiltonians so that Green's functions for arbitrary systems can be determined. Some methods reduce a Hamiltonian to an impurity problem and using the quantum computer for the solution of the impurity problem may be advantageous if only to decrease the number of operations required \cite{dallaire2016method,dallaire2016quantum,bauer2016hybrid}. However, in those previous algorithms for impurity problems (and many other publications), the determination of expectation values ({\it e.g.}, $\langle\Psi|(\hat c_{i\sigma}+\hat c_{i\sigma}^\dagger)|\Psi\rangle$ or some other measurement that is not an eigenstate) will require several measurements to acquire results in general  and therefore contribute to even more wavefunction preparations. Those algorithms will not be feasible without an improved wavefunction method at scale, which most likely can not be done in general due to the complexity of the problem \cite{schuch2009computational}.

With regards to ground state wavefunction preparation with Lanczos routines, some previous work by Motta, {\it et.~al.} in Ref.~\onlinecite{motta2020determining} focused on using a Lanczos algorithm to find the wavefunction with even elements of the Krylov basis.  The implementation here is different since the full or truncated set of Lanczos coefficients is derived from state-preserving quantum counting, all Lanczos basis states are used instead of just even functions, and there is no use of imaginary time evolution.

\section{Conclusion}

Constructing the continued fraction representation of an arbitrary Green's function by implementing Lanczos recursion on a quantum computer was shown to efficiently generate all the coefficients necessary with a quantum counting algorithm. This method is an exponential improvement in terms of memory storage compared to classical methods, and the method preserves the wavefunction for the next computation. This can then be used to find another element of the interacting Green's function or another ground state nearby can be sought from the closer starting state. The ground state is also determinable with this method, provided an initial state is prepared from a known Hamiltonian.

\section{Acknowledgements}

The author is grateful for conversation with Anirban N.~Chowdhury, David Sénéchal, Alexandre Foley, Christopher Chubb, Guillaume Duclos-Cianci, Maxime Tremblay, and David Poulin. Funding for this work was provided by the postdoctoral fellowship from Institut quantique and Institut transdisciplinaire d'information quantique (INTRIQ). This research was undertaken thanks in part to funding from the Canada First Research Excellence Fund (CFREF).

The author is grateful to the US-UK Fulbright Commission for financial support under the Fulbright U.S. Scholarship programme as hosted by the University of York.  This research was undertaken in part thanks to funding from the Bureau of Education and Cultural Affairs of the United States Department of State.

\begin{appendix}

\section{Quantum Lanczos recursion}\label{lanczoscoeffex}

This section goes into more detail with the algorithm presented in the main text.  These steps mainly follow the presentation of the Supplemental Material for Ref.~\onlinecite{temme2011quantum} but presents a few examples for the present case of interest. Unlike the main text where operators $\hat{\mathcal{U}}$ were used to denote the unitary form of an operator, whenever an operator is applied here, it will be assumed that this can be done with unitary operators to keep the notation simple.

The starting point begins after a ground state wavefunction is found and then applying quantum phase estimation (QPE) to determine the energy $E$.  The resulting wavefunction on the quantum computer, $|\varphi\rangle$, after the two initial steps is
\begin{equation}
|\varphi\rangle=|\Psi\rangle|E\rangle|0\rangle|0\rangle.
\end{equation}

At the very first step, the $n=0$ term of the Lanczos recursion must be applied.  For this first step, $\beta_{0}$ is taken to be zero, so all that must be found is $\alpha_0=\langle\psi_0|\mathcal{H}|\psi_0\rangle$ with $|\psi_0\rangle=\hat\Omega|\Psi\rangle$. 

The next action is to form $|\psi_0\rangle=\hat\Omega|\Psi\rangle$,
\begin{equation}
\hat\Omega|\varphi\rangle=|\psi_0\rangle|E\rangle|0\rangle|0\rangle.
\end{equation}
Explicitly, it is assumed that the operation of $\hat\Omega$ is well defined in the sense that a unitary is applied according to the discussion in the main text for $\hat\Omega$. If the determinant of $\hat\Omega$ is not one, then the operator can be divided by an integer to ensure the application of the operator will return a resulting probability on the interval [0,1]. Shifts by constants may also be necessary depending on the operator requested.  All that needs to be kept track of at the end of the computation is that the final result must be unnormalized and un-shifted to obtain the true expectation value.

The quantum counting procedure now begins:
\begin{enumerate}
\item The Hamiltonian is properly modified and applied to obtain
\begin{align}
\mathcal{H}\hat\Omega|\varphi\rangle&=\mathcal{H}|\psi_0\rangle|E\rangle|0\rangle|0\rangle\\
&=\Big(\alpha_0|\psi_0\rangle+\alpha^\perp|\psi_0^\perp\rangle\Big)|E\rangle|0\rangle|0\rangle.
\end{align}
Each term of the Hamiltonian is written as a series of multi-qubit operators and also noting that the requirement of Hermiticity of the Hamiltonian ensures that each term's Hermitian conjugate also appears in the Hamiltonian. In the main text, care was taken to apply a unitary representing all operators with a symbol $\mathcal{U}$, but this will be omitted here for clarity.  It is sufficient for the purposes of this paper to say that existing methods accomplish the application of an operator to a wavefunction controllably without excessive overhead, making the application of $\mathcal{H}$ or a similar operator to $\Psi$ feasible. However, for review, some methods are discussed in Apx.~\ref{oblivious}.

The Hamiltonian operator is normalized so the operator has norm 1, similar to the treatment of  $\hat\Omega$. The result from the normalized unitary will be bound to the interval [0,1] representing a probability (the accepted counts divided by the total number of times that the algorithm was run) which can then be used to recover the proper value by multiplying the result by the normalization factor that was used. To ensure a positive result from the Hamiltonian, the external potential can be shifted by a constant. Shifting the potential in a Hamiltonian by a constant factor does not change the eigenvectors. 

\item The QPE algorithm \cite{nielsen2010quantum} is used here again to find the energy of the new state on the first register, but the result is contained in a superposition giving
\begin{align}
&\hat{\mathcal{U}}_\mathrm{QPE}\hat\Omega^\dagger\mathcal{H}\hat\Omega|\varphi\rangle=\nonumber\\&\Big(\alpha_0|\Psi\rangle|E\rangle+\alpha^\perp|\Psi^\perp\rangle|E^\perp\rangle\Big)|E\rangle|0\rangle
\end{align}
The amplitudes of the two superpositions match because each element of the superposition had the same operation applied.  So, the energy corresponding to the appropriate wavefunction ($\psi_0$ or $\psi_0^\perp$) with probability $\alpha_0$ or $\alpha^\perp$.  For the sake of conciseness, the energy resulting from this step ($E$ or $E^\perp$) is referred to as $E'$.
\item The operator $\hat\Omega^\dagger$ is applied to revert the state $|\psi_0\rangle$ back to the original ground state wavefunction $\Psi$ in anticipation of a QPE being taken.
\item Through a series of bitwise operations provided by a CNOT gate (see Fig.~\ref{CNOT}), the energies in both registers are converted to a single comparison bit with value of 0 (matching) or 1 (not matching).  The output of the CNOT gate results in 
a single qubit represents the matching or not-matching result.

\begin{figure}[b]
\centering

\begin{minipage}{0.08\textwidth}
\centering
\begin{tabular}{|c|c||c|}
\hline
$a$ & $b$ & $a\oplus b$\\
\hline
0 & 0 & 0\\
0 & 1 & 1\\
1 & 0 & 1\\
1 & 1 & 0\\
\hline
\end{tabular}
\end{minipage}
\begin{minipage}{0.39\textwidth}
\includegraphics[width=0.6\columnwidth]{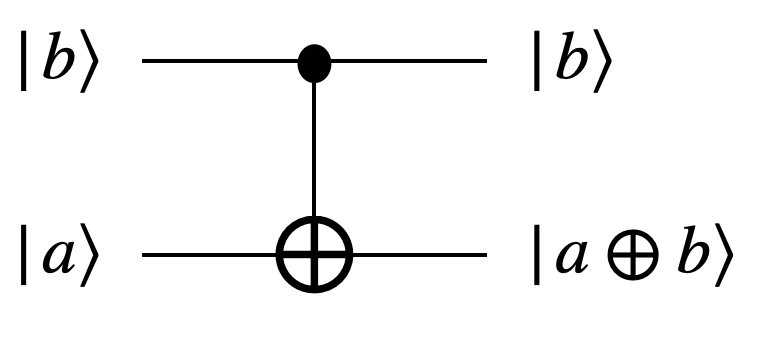}
\end{minipage}

\caption{Controlled-NOT (CNOT) gate taking inputs $a$ and $b$ and returning $a\oplus b$ \cite{nielsen2010quantum}. A table of shown on the left for all types of binary inputs and the schematic is shown on the right. This gate provides a comparator between input bits on the $a$ and $b$ valued qubits.  \label{CNOT}}

\end{figure}

The wavefunction is now
\begin{align}
&\hat{\mathcal{U}}_\mathrm{CNOT}\hat{\mathcal{U}}_\mathrm{QPE}\hat\Omega^\dagger\mathcal{H}\hat\Omega|\varphi\rangle=\nonumber\\&\Big(\alpha_0|\Psi\rangle|E\rangle|0\rangle+\alpha^\perp|\Psi^\perp\rangle|E^\perp\rangle|1\rangle\Big)|E\rangle
\end{align}
\item The single comparison qubit (also called a pointer qubit) is measured.  Measuring only one bit avoids a full collapse of the wavefunction, which is not advantageous here since this means $\psi_0$ would need to be completely re-prepared, which is the main issue that the algorithm avoids.

The algorithm then forks with two possibilities.
\item (Accept) If a value of 0 is measured, then the original wavefunction is recovered and the algorithm can continue at step 1.   A counter is incremented by one.  
\item (Reject) If a value of 1 is measured on the pointer qubit, then a state in the set $\psi_0^\perp$ was recovered.  A recovery procedure that is explained below is then performed to recover $\psi_0$ before returning to step 1.
\end{enumerate}

The ratio of successful iterations of the algorithm to the total number of times that the algorithm is run approximates the coefficient $\alpha_0$.

\subsection{Recovery of the original state}\label{recoveryexplain}

On some applications of the counting algorithm, the wrong energy ($E^\perp$) will be found.  The basic idea of the steps necessary to recover the original wavefunction is that the operators are undone and repeated until the original wavefunction is recovered.

If the result of the pointer qubit is found to be $|1\rangle$, then the resulting wavefunction is 
\begin{equation}
|\Psi^\perp\rangle|E\rangle|E^\perp\rangle|1\rangle.
\end{equation}
The main focus of the recovery operation is on the $|\Psi^\perp\rangle|E^\perp\rangle$ parts of the state.

First, the QPE resulting in $E^\perp$ is undone by applying the adjoint of all unitary operations.  The operators (for example, $\hat G_n$ from the main text) are then applied again to the state $\Psi^\perp$ and then another QPE is taken.  This gives back the state
\begin{equation}
\Big(\gamma|\Psi\rangle|E\rangle|0\rangle+\delta|\Psi^\perp\rangle|E^\perp\rangle|1\rangle\Big)|E\rangle
\end{equation}
after applying the comparator CNOT operations. The pointer qubit is measured.  The process then repeats until the pointer qubit is measured to be $0$, signifying a match between $E$ and $E'$.  This completes the recovery procedure since $\Psi$ is recovered. The algorithm then returns to step 1.

\subsubsection{Precision required}

Measuring the lone pointer qubit introduces some error into the wavefunction. By only requesting the first $d$ digits of the energies match, two goals are accomplished. First, the errors from measurement are contained to only a few digits. Second, the QPE is faster since it involves fewer operations when it is evaluated out to fewer digits of precision. As a consequence of the first point, the number of digits should prevent us from overlapping with excited states. 

\subsubsection{Cases with degeneracy}

As pointed out in the Supplemental Material of Ref.~\onlinecite{temme2011quantum}, the case where a degeneracy appears must be resolved with several QPE steps.  Particularly, $\nu$ processes are run, outputting to $\nu+1$ pointer qubits.  The average energy of this ensemble converges to the average energy requested out to the precision sought. Note that this would be useful for analyzing thermal states which is in a similar spirit to the Monte Carlo sampling in Ref.~\onlinecite{temme2011quantum}.

\subsubsection{Number of iterations to recover the wavefunction}

There is a question of how many iterations are required in a given recovery procedure to obtain $|\Psi\rangle$ again having projected the state into $|\Psi^\perp\rangle$.  This was answered in considerable detail in Ref.~\onlinecite{temme2011quantum}, and the same mapping onto the positive projection operators can be used to obtain
\begin{equation}
p_\mathrm{fail}\lesssim1/({2e(k+1)})
\end{equation}
for $k$ iterations of the recovery algorithm and a probability of failure of $p_\mathrm{fail}$.  The number of iterations can be increased until recovery here.

\subsection{Higher order terms in the recursion}

The previous section demonstrated the method for the $n=0$ term.  For $n>0$, the recursion relation can be exploited to give a single operator acting on $|\psi_0\rangle$.  The first few are listed here:
\begin{align}\label{alphaforms1}
|\psi_1\rangle&=\left(\mathcal{H}-\alpha_0\right)\hat\Omega|\Psi\rangle\\
\label{alphaforms2}
|\psi_2\rangle&=\left(\mathcal{H}-\alpha_1\right)|\psi_1\rangle-\beta_1|\psi_0\rangle\\
&=\Big(\left(\mathcal{H}-\alpha_1\right)\left(\mathcal{H}-\alpha_0\right)-\beta_1\Big)\hat \Omega|\Psi\rangle \nonumber\\
\label{alphaforms3}
|\psi_3\rangle&=\left(\mathcal{H}-\alpha_2\right)|\psi_2\rangle-\beta_2|\psi_1\rangle\\
&=\Big[\left(\mathcal{H}-\alpha_2\right)\Big(\left(\mathcal{H}-\alpha_1\right)\left(\mathcal{H}-\alpha_0\right)-\beta_1\Big) \nonumber\\
&\quad-\beta_2\left(\mathcal{H}-\alpha_0\right)\Big]\hat\Omega|\Psi\rangle \nonumber
\end{align}
Again, each operator can be normalized and then applied as discussed previously.  Note that only one wavefunction was necessary at each step but that the operator changes.

\subsection{Determination of $\beta_n$ coefficients}

Once $\alpha_n$ is determined, the $\beta_n$ coefficient must be determined. The operators necessary for the $\beta_n$ coefficients follow a similar form except instead of taking the conjugate of the operator listed in Eq.~(\ref{alphaforms1}--\ref{alphaforms3}) or a higher order term, the operator listed (for the $n$th term) and the previous operator ($n-1$ term) is used. 

To illustrate this method, the $\beta_1$ coefficient will be demonstrated. As pointed out in the text from a straightforward utilization of the Lanczos recursion relation, the $\beta_n$ coefficients are
\begin{equation}\label{betacoeff}
\beta_n=\langle\psi_{n-1}|\mathcal{H}|\psi_n\rangle=\langle\psi_n|\mathcal{H}|\psi_{n-1}\rangle.
\end{equation}
For the case of $\beta_1$, the operator $\left(\mathcal{H}-\alpha_0\right)\hat\Omega$ is applied to $|\Psi\rangle$ with the additional $\mathcal{H}$ applied as well, finally giving
\begin{align}
\beta_1&=\langle\psi_0|\mathcal{H}(\mathcal{H}-\alpha_0)\hat\Omega|\Psi\rangle\\
&=\langle\Psi|\hat\Omega^\dagger\mathcal{H}(\mathcal{H}-\alpha_0)\hat\Omega|\Psi\rangle.
\end{align}
Thus, for $\beta_1$, the operator $\hat\Omega^\dagger\mathcal{H}(\mathcal{H}-\alpha_0)\hat\Omega$ must be applied under the same constraints as before to discover the coefficient.  

\subsection{Application of operators to wavefunctions}\label{oblivious}

In the discussion so far, one step was omitted because it has appeared in several places already in the literature (see for example Refs.~\onlinecite{low2019hamiltonian,childs2017quantum,paetznick2013repeat}). Some proposals to apply operators onto the wavefunction are reviewed here.

A generic operator $\hat M$ can be expressed as a linear combination of unitary operators \cite{low2019hamiltonian,childs2017quantum}
\begin{equation}\label{LCU}
\hat M=\sum_{\ell=1}^\tau q_\ell\hat U_\ell
\end{equation}
for some set of unitaries $\hat U_\ell$ with coefficients $q_\ell$. For an example of how operators can be applied to wavefunctions the repeat-until-success algorithm \cite{paetznick2013repeat} implements a unitary operator controlled on $\ell$ auxiliary qubits, one for each term in the expansion of Eq.~\eqref{LCU}. Formally, applying an operator $\hat M$ onto a wavefunction $\Psi$ gives \cite{paetznick2013repeat}
\begin{equation}
\hat M|\Psi\rangle|0\rangle^{\otimes \ell}=\rho|\hat M\Psi\rangle|0\rangle^{\otimes \ell}+\rho^\perp|(\hat M\Psi)^\perp\rangle |q\rangle^{\otimes \ell}
\end{equation}
which has a similar form to Eq.~\eqref{HamPsi} but where $\ell$ auxiliary qubits were required to implement the unitary $\hat\Omega$ and $q$ is a value 0 or 1 but not all uniform 0. The unitaries are provided one auxiliary qubit each in the linear combination of unitaries representing $\hat M$. Measuring those $\ell$ (pointer) qubits determines whether the correct operator was applied (all 0s, for example) or if the resulting state resides in the perpendicular space (not all 0s). Otherwise, just as in the quantum counting algorithm from Apx.~\ref{recoveryexplain}, a recovery procedure is used to obtain the initial state again only to reapply the operators until success is found.

The probability of obtaining the correct state upon measurement can be improved by using oblivious amplitude amplification (OAA) \cite{kothari2014efficient}, a similar method to quantum amplitude amplification \cite{brassard1998quantum} or Grover's search \cite{grover1996fast}. The state $|0\rangle^{\otimes \ell}$ is the target state sought in OAA, hence the state remains `oblivious' to the input wavefunction \cite{kothari2014efficient,berry2015simulating}. Doing this procedure increases the probability of obtaining the correct state to near certainty.

\section{Overview of quantum counting}\label{history}

The origins of quantum counting for quantum algorithms will be discussed here in the context of Refs.~\onlinecite{brassard1998quantum,brassard2002quantum,marriott2005quantum,nagaj2009fast,temme2011quantum} and the Supplemental Material of Ref.~\onlinecite{temme2011quantum}. The main reason for this review of the history of these methods is that some names of the methods have blended together through time and it is necessary to contrast the algorithm here with other algorithms with similar sounding names. The use of the name ``quantum counting" here matches the colloquially used name of this algorithm, particularly by one of the authors of Ref.~\onlinecite{temme2011quantum} \cite{davidPconv}. This appendix reviews the quantum algorithms and shows that the concepts behind amplitude estimation, counting, and QMA-sampling are very similar.

\paragraph*{Quantum counting.--}An algorithm titled `quantum counting' was introduced in Ref.~\onlinecite{brassard1998quantum}.  This algorithm counts the number of elements in a superposition that satisfies some criterion. The algorithm is similar to Grover's algorithm with one or more states sought. 

\paragraph*{Quantum amplitude estimation.--}It is remarked by Brassard, {\it et.~al.} in Ref.~\onlinecite{brassard1998quantum} that the quantum counting algorithm described above is related to the general method of quantum amplitude estimation (QAE).  QAE is defined in Ref.~\onlinecite{brassard2002quantum} as the process by which a state (unnormalized) that is separated into ``good" ($\Psi_1$) and ``bad" ($\Psi_0$) states
\begin{equation}\label{QAE1}
|\Psi\rangle=|\Psi_1\rangle+|\Psi_0\rangle
\end{equation}
can identify the amplitude $\mathcal{A}=\langle\Psi_1|\Psi_1\rangle$. Note the similarity of this form to the target and perpendicular ($\perp$) states in the main text.

\paragraph*{QMA-amplification (state-preserving quantum counting).--} In 2005, Marriott and Watrous introduced QMA-amplification \cite{marriott2005quantum}. QMA-amplification is the strategy that is used here that is different than quantum counting with Grover's search, but the algorithm still determines the expectation value by counting the number of transitions from a state back to that same state on application of an operator.   The essential idea is the same as the algorithm in this paper and the methods there are the direct inspiration here.

Notably, the QMA-amplification algorithm was used to implement a quantum Metropolis sampling algorithm in Ref.~\onlinecite{temme2011quantum}, essentially replacing the copy operation required in classical implementation of Monte Carlo sampling.  The distinction made in this paper (and others \cite{nagaj2009fast}) is to also identify the Watrous and Marriott method as state-preserving quantum counting. In this case, counting transitions is not exactly the same as counting the number of states in a superposition, but both algorithms rely on counting with quantum physics, hence the names.

Summarily, ``quantum counting" from Ref.~\onlinecite{brassard1998quantum} counts the number of states satisfying some criterion in a superposition. Meanwhile, ``quantum counting" in the sense of QMA-amplification counts the number of times a state is recovered when applying an operator \cite{marriott2005quantum,temme2011quantum}. The core difference between the two methods is that the amplitude is interpreted as a weight of a given state instead of a sum of several states. The algorithm from Ref.~\onlinecite{marriott2005quantum} (and particularly the use in Ref.~\onlinecite{temme2011quantum}) will be referred to as ``state-preserving quantum counting" or just quantum counting when there is no potential for ambiguity.  

\subsection{Use of quantum counting for quantum chemistry}

As to the application of quantum counting in quantum chemistry, it was remarked in the Supplemental Information of Ref.~\onlinecite{temme2011quantum} that the counting algorithm can be used as an oracle query to find derivatives of the energy to relax molecular structures \cite{kassal2009quantum}.  The oracle query is needed for the quantum gradient algorithm of Ref.~\onlinecite{jordan2005fast}. So, the method has been proposed for quantum chemistry previously \cite{BP20}, but it is shown to be strongly applicable to Lanczos here.

\end{appendix}

\bibliography{QLR}

\end{document}